\documentclass[letter,12pt]{article}
\usepackage{amsmath}
\usepackage{enumerate}
\usepackage{amssymb}
\usepackage{geometry}
\usepackage{setspace}
\usepackage{graphicx}
\usepackage{natbib}
\usepackage{lmodern}
\usepackage{soul}
\usepackage{comment}
\usepackage{ifthen}
\usepackage{xr}
\usepackage{framed}
\usepackage{algorithm}
\usepackage{algpseudocode}
\newtheorem{theorem}{Theorem}

\newtheorem{lemma}{Lemma}
\newtheorem{proposition}{Proposition}

\newenvironment{proof}[1][Proof]{\textbf{#1.} }{\ \rule{0.5em}{0.5em}}
\DeclareMathSizes{12}{11}{7.7}{5.5}
\usepackage[inline]{enumitem}
\usepackage{comment}

\newcommand{\RomanNumeralCaps}[1]
    {\MakeUppercase{\romannumeral #1}}

\graphicspath{ {WD-Graphs/} }

\begin{document}

\title{Credit Valuation Adjustment with Replacement Closeout: Theory and Algorithms}


\thispagestyle{empty}
\author{Chaofan Sun\and Ken Seng Tan\and Wei Wei %
\thanks{Sun: Department of Actuarial Mathematics and Statistics, School of Mathematics and Computer Science, Heriot-Watt University, Edinburgh, Scotland, EH14, 4AS, UK. E-mail: cs57@hw.ac.uk.  %
Tan: Division of Banking \& Finance, Nanyang Business School, Nanyang Technological University,
Singapore. Email: kenseng.tan@ntu.edu.sg.
Wei: Department of Actuarial Mathematics and Statistics, School of Mathematics and Computer Science, Heriot-Watt University, Edinburgh, Scotland, EH14, 4AS, UK. E-mail: wei.wei@hw.ac.uk. 
}
}
%
%
\thispagestyle{empty}
\date{\today}
\maketitle
\thispagestyle{empty}

\vspace*{-0.4cm}
\begin{abstract}
The replacement closeout convention has drawn more and more attention since the 2008 financial crisis. Compared with the conventional risk-free closeout, the replacement closeout convention incorporates the creditworthiness of the counterparty and thus providing a more accurate estimate of the Mark-to-market value of a financial claim. In contrast to the risk-free closeout, the replacement closeout renders a nonlinear valuation system, which constitutes the major difficulty in the valuation of the counterparty credit risk.\par
In this paper, we show how to address the nonlinearity attributed to the replacement closeout in the theoretical and computational analysis. In the theoretical part, we prove the unique solvability of the nonlinear valuation system and study the impact of the replacement closeout on the credit valuation adjustment. In the computational part, we propose a neural network-based algorithm for solving the (high dimensional) nonlinear valuation system and effectively alleviating the curse of dimensionality. We numerically compare the computational cost for the valuations with risk-free and replacement closeouts. The numerical tests confirm both the accuracy and the computational efficiency of our proposed algorithm for the valuation of the replacement closeout.

\end{abstract}


\newpage
\setcounter{page}{1}
\section{Introduction}\label{SecIntroduction}
The purpose of this paper is to provide theoretical and computational underpinnings for the application of the replacement closeout convention in the valuation of counterparty credit risk. Credit valuation adjustment (CVA), commonly viewed as the value of counterparty credit risk, has become an essential risk measure in the financial risk management since the 2008 global financial crisis. A challenging task in the calculation of CVA is modelling and calculating the Mark-to-market (MtM) closeout amount when a  default event is triggered. Before the 2008 financial crisis, almost all credit risk models were based on the so called risk-free closeout convention in the sense that the counterparty is assumed to be risk-free and thus the expectation of the discounted future cash flows is used in recovery calculation\footnote{See for example \cite{jarrow2001counterparty}, \cite{blanchet2004hazard} \cite{bielecki2008pricing}, \cite{BrigoChourdakisForth} and \cite{brigo2009counterparty} for the valuation of defaultable financial claims under the risk-free closeout convention. }. \par


A painful lesson of the global financial crisis is that no counterparty should be considered risk-free and banks must mark to market as much as possible\footnote{See Chapter CAP50 of Basel \RomanNumeralCaps{3}.}. Due to the inadequate consideration for the credit quality of the counterparty, the risk-free closeout fails to adequately reflect the MtM value of a financial claim upon the default time. In the post-crisis period, many studies have questioned the validity of the risk-free closeout (e.g., \citeauthor{burgard2011partial} \citeyear{burgard2011partial}, \citeauthor{brigo2011close} \citeyear{brigo2011close} and \citeauthor{brigo2012counterparty} \citeyear{brigo2012counterparty}.). In particular, the International Swaps and Derivatives Association (ISDA) highlights the importance of incorporating the creditworthiness of the counterparty in closeout, and thus suggesting an alternative way of calculating the recovery.\footnote{Quoting from ISDA (2010): ``Upon default close-out, valuations will in many circumstances reflect the replacement cost of
transactions calculated at the terminating party's bid or offer side of the market, and will often take into account the creditworthiness of the terminating party.''}\par

To overcome the drawback of the risk-free closeout convention, more and more attention has been devoted to the replacement closeout over the past decade.\footnote{See for example \cite{crepey2010counterparty}, \cite{hu2012fully}, \cite{henry2012cutting}, \cite{brigo2016nonlinearity} and \cite{biagini2019pricing}.}
The idea of working with the pre-default value in the recovery calculation is first suggested in \cite{duffie1996swap} though they do not use the term ``replacement closeout". As opposed to the risk-free closeout, the replacement closeout is defined via the pre-default value of a financial claim. As the  counterparty credit risk is inherently embedded with the pre-default value, the replacement closeout
is more appropriate than the risk-free closeout since it takes into consideration the counterparty's creditworthiness so that
the MtM value of the contract is adequately captured and thus yields a more accurate estimate of CVA. The appropriateness of replacement closeout, however, is achieved at the expense of valuation complexity. This is attributed by the fact that one cannot obtain the pre-default value without knowing the future pre-default value and vice versa, thus the definition of the pre-default value under the replacement closeout is actually a loop rather than an explicit representation. Consequently, the replacement closeout creates a nonlinear system, which poses the major difficulty in the valuation.\par


In this paper, we aim to contribute to the literature on replacement closeout analysis, particularly in addressing the nonlinearity, from two aspects: theoretically and computationally.
In the theoretical analysis, we produce two theoretical results. First, we prove that the nonlinear valuation system arising from the replacement closeout admits a unique solution in broad generality. Second, we show that the CVA calculated from the replacement closeout is always larger than the corresponding value from the risk-free closeout. This result suggests that without taking into consideration the counterparty's creditworthiness in recovery calculation would underestimate the CVA. More severely, through an example of a defaultable European option, we find that the underestimation could be significant, especially for unhealthy counterparty's creditworthiness. Together with the example, our theoretical result unequivocally points out
the importance of tracking the MtM value of financial claims upon the default times and cautions against ignoring the creditworthiness of the counterparty in the closeout. It is noteworthy that the closeout function considered in our model is not only restricted to CVA models, but also nests most of models of debit valuation adjustment (DVA) and funding valuation adjustment (FVA). Therefore all the results aforementioned also apply to the analysis of valuation adjustment (XVA)\footnote{For the studies about DVA, FVA and XVA, see for example \cite{pallavicini2011funding}, \cite{albanese2015fva}, \cite{hull2016xvas} and \cite{andersen2019funding}.}.
\par

From a practitioner's perspective, the applicability of the replacement closeout largely hinges on its computational complexity. The nonlinearity of the replacement closeout gives rise to significant computational burdensome. Particularly, when the number of risk factors\footnote{The sources of risk factors would be macroeconomic factors (e.g., interest rates), market indices (e.g., S\&P 500 index), the credit ratings of the investors and their counterparties, and the randomness of assets (e.g. multi-asset, or stochastic volatility models), etc.} increases, the computational complexity of the nonlinear valuation system  grows exponentially. As the number of risk factors usually indicates the number of dimensions, the exponential growth on the computational cost is also called the ``curse of dimensionality'' \footnote{See, e.g., \cite{bellman1957dynamic}.}. Recently,  the deep backward stochastic differential equation (BSDE) method, advocated by \cite{weinan2017deep}, \cite{han2018solving} and \cite{beck2019machine}, has demonstrated a promising performance in overcoming the curse of dimensionality for high dimensional dynamic models. In the Deep BSDE method, the dynamic model (usually a partial differential equation (PDE) problem) is reformulated into a BSDE problem. Then, applying a standard simulation and discretization to the BSDE problem, one can obtain an optimization problem associated with the terminal value. Finally, neural networks are employed to approximate the unknown gradients of the optimization problem. \par
The Deep BSDE method has been applied to risk management\footnote{Besides the Deep BSDE method, there are also several other attempts to approximate the high dimensional dynamic models in various areas \citep[e.g.,][]{carleo2017solving,sirignano2018dgm, weinan2018deep, weinan2019multilevel}. Readers interested in this direction would be referred to the literature surveys by \cite{ruf2020neural} and \cite{beck2020overview}.}.  Building on \cite{weinan2017deep}, \cite{henry2017deep} develops a primal-dual method for the valuation of counterparty credit risk. Combining the dynamic programming, \cite{hure2019some} develop a new Deep BSDE solver which targets at  approximating the value function.  Based on the Deep BSDE method, \cite{gnoatto2020deep} design an XVA solver and deduce a posteriori bounds on the error of the neural network approximations. \cite{albanese2021xva}, with a BSDE formulation and GPU computing, propose a deep learning regression method for the credit risk valuation.
In the existing Deep BSDE solvers for credit risk valuation, the decision variables, usually the value functions and (or) their gradients, are viewed as a sequence of functions of state variables, and each function at a time node is attached to a fully connected (FC) neural network. For ease of reference, we refer to such a Deep BSDE solver as a multi-FC Deep BSDE solver. 
\cite{raissi2017physics}, \cite{raissi2018deep} and \cite{raissi2018forward}, on the other hand, develop a single network-based Deep BSDE method. Because of the dramatic decrease of the number of neural networks, \citeauthor{raissi2017physics}'s method allows more delicate neural network architectures to be applied in the Deep BSDE solver. From the reported numerical experiments in  \cite{chan2019machine}, one can see that the multi-FC Deep BSDE solver could be significantly improved by the single network-based method with a proper neural network.\par

 Motivated by  \cite{raissi2018forward}, the computational contribution of this paper is to develop a Deep BSDE algorithm for the calculation of CVA with replacement closeout. In contrast to the existing literature on the valuation of credit risk via the Deep BSDE method, our algorithm is single network-based. Moreover, instead of the fully connected neural network used in \cite{weinan2017deep}, we approximate the gradient function by the long-short term memory (LSTM) neural network, which is originally
 proposed by \cite{hochreiter1997long} and has been shown to have impressive power in solving real-world problems with sequential data\footnote{For example, the LSTM neural network has been successfully used in time series prediction \citep{schmidhuber2005evolino}, speech recognition \citep{graves2013speech} and rhythm learning \citep{gers2002learning}. }. The numerical tests suggest that our proposed solver produces satisfactory efficiency and yields better convergence stability than the multi-FC Deep BSDE solver.
Finally, we compare the computational cost for calculating CVA under both risk-free and replacement closeout. We test our algorithm on the valuations of a defaulable claim with risk-free and replacement closeout respectively.  Because of the computational efficiency of the single network-based algorithm in calculating the replacement closeout, it only requires about half of the computational time compared to that with the risk-free closeout. This therefore effectively alleviates the curse of dimensionality and enhances the practical application of the replacement closeout.
\par

The remainder of the paper is organized as follows. In Section \ref{GeneralModel}, we formulate the valuation of a generic defaultable financial claim with replacement closeout. In Section \ref{Sec:Valuation}, we prove that the valuation equation admits a unique solution. The proof is based on an iteration, which asserts that the conventional risk-free closeout would underestimate CVA. In Section \ref{Sec:Algorithm}, we describe the single network-based algorithm. Numerical examples and discussions are provided in Section \ref{sec:NRD}. Section \ref{Sec:Conclusions} concludes the paper and lemmas for the proof of main results are relegated to the appendix.

\section{General Setup }\label{GeneralModel}
\subsection{The dynamics}
We start with a filtered probability space $(\Omega,\mathcal{F},\{\mathcal{F}_t\}_{t\ge0},\mathbb{P})$ which satisfies the usual conditions. Let $W= (W_1,W_2,\ldots,W_n)^T$ define an $n$-dimensional Brownian motion adapted to the filtration $\{\mathcal{F}_t\}_{t\ge0}.$  Filtration $\{\mathcal{F}_t\}_{t\ge0}$ captures all accessible information generated by a family of Markov process $X=X^x$ parameterized by the initial state $X_0=x\in\mathbb{R}^m$ and governed by the following
stochastic differential equation (SDE),
\begin{equation}\label{SDE}
dX_t = \mu(t,X_t) dt + \sigma(t,X_t) dW_t
\end{equation}
where  $\mu(t,x) = (\mu_1(t,x),\mu_2(t,x),\ldots,\mu_m(t,x))^T$, $\sigma(t,x) = (\sigma_{ij}(t,x))_{i=1,2,\ldots,m, j=1,2,\ldots,n}$ are functions defined on $(\mathbb{R}_+,\mathbb{R}^m)$ and valued in $\mathbb{R}^m$ and $\mathbb{R}^{m\times n}$ respectively.\par
Suppose that $\mu$ and $\sigma$ are Lipschitz continuous functions, i.e., there exists an $L>0$ such that for  $x\neq y, t\neq s$,
\begin{align}\label{Lip}
\vert g(t,x) - g(s,y)\vert\le L(\vert t-s \vert+\vert x - y\vert),
\end{align}
where $g\in\{\mu, \sigma\}, \vert A\vert: = \sqrt{\sum_{j=1}^m\sum_{i=1}^n A_{ij}^2}, \forall A\in \mathbb{R}^{n\times m}, n, m \in\mathbb{N}_+.$ \par
Condition (\ref{Lip}) guarantees the existence and uniqueness of the strong solution to SDE (\ref{SDE}). Moreover, it follows from standard arguments in SDE \citep[e.g.][chapter 1]{pham2009continuous} that, for any $T>0, p>1$, there exists a constant $L$, such that
\begin{align}\label{Moment}
\mathbb{E}\left[\sup_{t\in[0,T]}\vert X_t \vert ^p \right]\le L(\vert x\vert^p+1),
\end{align}
where $L$ depends on $T$ and $p.$\par
Define the infinitesimal Markov generator associated with $X$ by
\begin{align*}
\mathcal{L}:=\frac{\partial}{\partial t} + \sum_{i=1}^m\sum_{j=1}^m a_{ij}(t,x)\frac{\partial^2}{\partial x_i\partial x_j}+\sum_{i=1}^m b_i(t,x)\frac{\partial}{\partial x_i},
\end{align*}
where $a_{ij}(t,x) =\frac{1}{2} \sum_{k=1}^n\sigma_{ik}(t,x)\sigma_{jk}(t,x)$.
Suppose that $\mathcal{L}$ is uniformly parabolic for any $T>0$, i.e., there exists a positive constant $\lambda_1>0$, such that for any $\xi=(\xi_1,\xi_2,\ldots,\xi_m)^T\in\mathbb{R}^m,$
\begin{align}\label{Parabolic}
\sum_{i=1}^m\sum_{j=1}^ma_{ij}(t,x)\xi_i\xi_j\ge \lambda_1\vert\xi\vert^2  , \forall (t,x)\in[0,T]\times\mathbb{R}.
\end{align}

\subsection{Default times and hazard rates}
Let $\tau$  be a non-negative random variable on $(\Omega,\mathcal{F},\mathbb{P})$ that captures the default time. Define the associated default indicator process by $H_t = 1_{\tau\le t}$. Let $\{\mathcal{H}_t\}_{t\ge 0}$ denote the filtration generated by $\{H_t\}_{t\ge 0}$. Suppose that all information in the market at time $t$, defined by $\sigma$-algebra $\mathcal{G}_t,$ comprises the accessible information represented by $\mathcal{F}_t$ and the information generated by the observation of the occurrence of the default up to $t$, i.e., $\mathcal{G}_t=\mathcal{F}_t\vee\mathcal{H}_t$.\par
Denote $G_t =\mathbb{P}(\tau>t\vert\mathcal{F}_t)$, so that $G_t$ represents the survival process of the default time $\tau$ with respect the reference filtration $\{\mathcal{F}_t\}_{t\ge 0}$. Define the hazard rate associated with $G_t$ by $\lambda_t$, e.g., $G_t = \exp\{-\int^t_0\lambda_sds\}$, where $\{\lambda_t\}_{t\ge0}$ is a non-negative progressively measurable process with respect to $\{\mathcal{F}_t\}_{t\ge0}$, with integrable sample paths. In order to keep the Markov property of the model, throughout the paper  we suppose that $\lambda_t = \lambda(t,X_t)$, where $\lambda$ is a Borel function from $\mathbb{R}_+\times \mathbb{R}^m$ to $\mathbb{R}_+.$
\subsection{Cash flows, pre-default values and CVA }
Consider a defaultable claim maturing at $T$. The holder of the claim receives a payment flow with rate $c_t$ until the default time $\tau$ or the maturity date $T$, whichever comes earlier. At the maturity date $T$, the holder receives the terminal payoff $\phi_T$ if the default event has not occurred yet. If there is a default before time $T$, then the holder receives a lump-sum payoff $Z_{\tau}$, referred to as the closeout amount, at the default time $\tau$. Denote the risk-free rate by $r_t$. Suppose that $c_t, \phi_T, Z_{\tau}, r_t$  are given by $c(t,X_t), \phi(X_T)$ $Z(\tau,X_{\tau}), r(t,X_t)$ respectively, where $c:\mathbb{R}_+\times\mathbb{R}^m\rightarrow\mathbb{R}, \phi: \mathbb{R}^m\rightarrow\mathbb{R}, Z:\mathbb{R}_+\times\mathbb{R}^m\rightarrow\mathbb{R}, r: \mathbb{R}_+\times\mathbb{R}^m\rightarrow\mathbb{R}_+$ are Borel functions.\par
The discounted payoff process of the defaultable claim is given by
\begin{align*}
P_t =\int^{\tau\wedge T}_te^{-\int^s_t r(u,X_u)du}c(s,X_s)ds + 1_{t<\tau\le T}Z(\tau,X_{\tau})e^{-\int^{\tau}_t r(s,X_s)ds} + \phi(X_T)e^{-\int^T_t r(s,X_s)ds} 1_{\tau>T}.
\end{align*}
Define the value process of the defaultable claim by $\mathbb{E}[P_t\vert\mathcal{G}_t]$, then it follows from the standard arguments in the literature on credit risk \citep[e.g.][chapter 8]{bielecki2013credit} that $\mathbb{E}[P_t\vert\mathcal{G}_t] = 1_{\tau>t}\bar{V}(t,X_t)$, where $\bar{V}$ is referred to as pre-default value function and is given by

\noindent $\bar{V}(t,x) = $
\begin{align}\label{PreDefaultValue}
\mathbb{E}\left[\int^T_t(c(s,X_s)+\lambda(s,X_s)Z(s,X_s))e^{-\int^s_tr(u,X_u)+\lambda(u,X_u)du}ds+e^{-\int^T_tr(s,X_s)+\lambda(s,X_s)ds}\phi(X_T)\vert X_t=x \right]
\end{align}
We suppose that $c, \lambda, Z, \phi, r$ have polynomial growth in $x$,  uniformly in $t$, i.e., there exist $L>0, n>0,$ such that
\begin{align}\label{PolynomialGrowth}
\vert g(t,x)\vert\le L(\vert x\vert ^n+1),\forall (t,x)\in[0,T]\times \mathbb{R}^m,
\end{align}
where $g\in\{c, \lambda, Z, \phi, r\}$ and $L,n$ depend on $T.$\par
The following proposition demonstrates that $\bar{V}$ has polynomial growth in $x$, uniformly in $t$.
\begin{proposition}\label{Prop:PolynomialGrowth}
 There exist $L>0, p>0$, such that
\[
\vert \bar{V}(t,x) \vert\le L(\vert x\vert ^p+1), \forall (t,x)\in[0,T]\times \mathbb{R}^m.
\]
\end{proposition}
\begin{proof}
The proof is just an immediate consequence of some algebra and we provide the proof for completeness.
It is easy to see that
\begin{align*}
\vert \bar{V}(t,x)\vert \le \mathbb{E}\left[\int^T_t(\vert c(s,X_s)\vert +\lambda(s,X_s)\vert Z(s,X_s) \vert )ds+\phi(X_T)\vert X_t=x \right].
\end{align*}
Then it follows from (\ref{PolynomialGrowth}) that there exist $L_0>0, p_0>0,$ such that
\begin{align*}
\vert \bar{V}(t,x)\vert& \le \mathbb{E}\left[\int^T_tL_0(\vert X_s\vert^{p_0}+1) ds+L_0(\vert X_T\vert^{p_0}+1)\right] \\& \le \mathbb{E}\left[TL_0\left(\sup_{0\le s\le T}\vert X_s\vert^{p_0}+1\right) +L_0(\vert X_T\vert^{p_0}+1)\vert X_t=x\right].
\end{align*}
The result is an immediate consequence of (\ref{Moment}).
\end{proof}\par
Moreover, to guarantee the $C^{1,2}$ regularity of $\bar{V}$,  we suppose that $c, \lambda, Z, \phi, r$ are continuous functions of $(t,x)$ and that $c, \lambda, Z, r$ are locally H{\"o}lder continuous in $x$, uniformly in $t$, i.e., for any compact set $E\subset[0,T]\times\mathbb{R}^m$, there exist $L>0, \alpha\in(0,1)$ such that
\begin{align*}
\vert g(t,x) -g(t,y)\vert\le L\vert x-y \vert^{\alpha}, \forall t, x, y\in E,
\end{align*}
where $g\in\{c, \lambda, Z, r\}$ and $L,\alpha$ depend on $E.$\par

By denoting $U(t,x)$ as  the risk-free value of the defaultable claim, then it can be obtained from (\ref{PreDefaultValue}) via
setting $\tau=\infty,$ or equivalently $\lambda(t,x)=0$; i.e.,
\begin{align}\label{RiskFreeValue}
U(t,x) = \mathbb{E}\left[\int^T_tc(s,X_s)e^{-\int^s_tr(u,X_u)du}ds+e^{-\int^T_tr(s,X_s)ds}\phi(X_T)\vert X_t=x \right].
\end{align}
Recall that CVA measures the value of credit risk. We define the CVA for a defaultable claim by the difference between the risk-free value and the pre-default value of the financial claim as
\begin{align}\label{CVA}
CVA:= U(t,x) -\bar{V}(t,x).
\end{align}
\section{Valuation}\label{Sec:Valuation}
\subsection{Closeout functions}
We denote the closeout payoff by a continuous function $f(t,x,y): \mathbb{R}_+\times\mathbb{R}^m\times\mathbb{R}\rightarrow\mathbb{R}$, where $t,x,y$ represent time, state value and the value of the claim respectively. We further assume that $f$ is locally H{\"o}lder continuous in $x$, uniformly in $t,y$.
To rule out the moral hazard in the financial claim, we suppose that
 the closeout payoff is no more than the value of the claim, i.e., $f(t,x,y)\le y, \forall (t,x,y)\in \mathbb{R}_+\times\mathbb{R}^m\times\mathbb{R}$ and that the closeout function $f$ and the loss function $y-f(t,x,y)$ are increasing in $y$, i.e.,
\begin{align}\label{MoralHazard}
0\le f(t,x,y_2)-f(t,x,y_1)\le y_2-y_1, \forall y_1<y_2, y_1,y_2\in\mathbb{R}, (t,x)\in \mathbb{R}_+\times\mathbb{R}^m.
\end{align}
Condition (\ref{MoralHazard}), which is referred to as incentive compatibility in the literature on (re)insurance,\footnote{See, for example, \cite{huberman1983optimal} and \cite{picard2000design}. Moreover, see \cite{chi2011optimal}, \cite{xu2019optimal} and \cite{tan2020optimal}, for the impact of condition (\ref{MoralHazard}) on the optimal (re)insurance contract design in static and dynamic models.}  links to the general revelation principle in economics\footnote{See, for example, \cite{myerson1979incentive} and \cite{dasgupta1979implementation}. }. This condition applies to most of closeout functions in credit modelling. In particular, it is easy to verify that all the closeout functions used in the calculations of CVA, DVA, FVA and XVA fall into this general category.
\subsection{Linear valuation with risk-free closeout}
As a benchmark, we briefly consider the case in which the risk-free value of the defaultable claim is used in the closeout function. We define the value of the defaultable claim with the risk-free closeout by $V_0$. The closeout function is then given by $f(t,x,U(t,x))$, where $U$ is the risk-free value of  the claim  given by (\ref{RiskFreeValue}).
Furthermore, $V_0$ can be obtained from (\ref{PreDefaultValue}) after the substitution
 $Z(t,x) = f(t,x,U(t,x))$; i.e.,
\begin{align}\label{PreDefaultValueRiskFree}
V_0(t,x) & = \mathbb{E} \left[\int^T_t(c(s,X_s)+\lambda(s,X_s)f(s,X_s,U(s,X_s)))e^{-\int^s_tr(u,X_u) +\lambda(u,X_u)du}ds \right. \nonumber\\& \left. \hspace*{0.4in} +e^{-\int^T_tr(s,X_s)+\lambda(s,X_s)ds}\phi(X_T)\vert X_t=x \right].
\end{align}
By denoting $\Pi_0(t,x)$ as the CVA calculated from the risk-free closeout, it follows from the definition of CVA (\ref{CVA}) that
\begin{align*}
\Pi_0(t,x)= U(t,x)-V_0(t,x).
\end{align*}
The following proposition demonstrates that the value of counterparty credit risk is positive.
\begin{proposition}\label{RiskFreeCVA}
$\Pi_0(t,x)\ge 0, \forall (t,x)\in[0,T]\times\mathbb{R}^m$.
\end{proposition}
\begin{proof}
As $c, \phi, r$ are locally H{\"o}lder continuous, $\{a_{ij}\}_{i,j\in\{1,2,\ldots,m\}}, b$ are Lipschitz continuous, $\mathcal{L}$ is uniformly parabolic and $U$ has polynomial growth (See Proposition \ref{Prop:PolynomialGrowth}.), then it follows from standard arguments in PDE\footnote{See, e.g., \cite{aronson1967parabolic}.} and the Feynman-Kac formula that $U\in C^{1,2}([0,T)\times\mathbb{R}^m)\cap C([0,T]\times\mathbb{R}^m)$ and solves the following Cauchy problem
\begin{align*}
\mathcal{L}U(t,x)-r(t,x)U(t,x)+c(t,x)&=0, \quad (t,x)\in[0,T)\times\mathbb{R},\\U(T,x) &= \phi(x), \quad x\in\mathbb{R}^m.
\end{align*}
Moreover, (\ref{MoralHazard}) implies that $f(t,x,y)$ has linear growth in $y$, uniformly in $(t,x)$ and thus together with the polynomial growth of $U$ in $x$, showing that $V_0$ has polynomial growth in $x$, $V_0\in C^{1,2}([0,T)\times\mathbb{R}^m)\cap C([0,T]\times\mathbb{R}^m)$ and satisfies
\begin{align*}
\mathcal{L}V_0(t,x)-(r(t,x)+\lambda(t,x))V_0(t,x)+c(t,x)+\lambda(t,x)f(t,x,U(t,x))&=0, \quad (t,x)\in[0,T)\times\mathbb{R}^m,\\V_0(T,x) &= \phi(x), \quad x\in\mathbb{R}^m
\end{align*}
As $\Pi_0 = U(t,x)-V_0(t,x)$, the difference of the above two PDEs yields
\begin{align*}
\mathcal{L}\Pi_0(t,x) - r(t,x)\Pi_0(t,x)+\lambda(t,x)(V_0(t,x)-f(t,x,U(t,x)))=0.
\end{align*}
As $f(t,x,U(t,x))\le U(t,x)$, we have that
\begin{align*}
\mathcal{L}\Pi_0(t,x) - r(t,x)\Pi_0(t,x)+\lambda(t,x)(V_0(t,x)-U(t,x))\le 0,
\end{align*}
which reduces to
\begin{align*}
\mathcal{L}\Pi_0(t,x) - (r(t,x)+\lambda(t,x))\Pi_0(t,x)\le 0.
\end{align*}
It is easy to see that $\Pi_0$ has polynomial growth and $\Pi_0(T,x)=0$, then it follows from the maximum principle\footnote{See, \citeauthor{friedman2008partial}, \citeyear{friedman2008partial}, chapter 2.} that $\Pi_0\ge 0$.\par
This completes the proof.
\end{proof}
\subsection{Nonlinear valuation with replacement closeout}
Under the replacement closeout convention, the closeout function $f$ depends on the pre-default value of the defaultable claim. Thus, the representation (\ref{PreDefaultValue}) of the pre-default value is no longer an explicit formula, rather it is a nonlinear equation of the pre-default value function.\par

Let  $V$ be the pre-default value obtained under the replacement closeout convention.
Then $V$ can be obtained from (\ref{PreDefaultValue}) after substituting
  $Z$ by $f(t,x,V(t,x))$; i.e.,
\begin{align}\label{PreDefaultEquation}
V(t,x) &= \mathbb{E}\left[\int^T_t(c(s,X_s)+\lambda(s,X_s)f(s,X_s,V(t,X_s)))e^{-\int^s_tr(u,X_u)+\lambda(u,X_u)du}ds \right. \nonumber\\ &\hspace*{0.4in} \left. +e^{-\int^T_tr(s,X_s)+\lambda(s,X_s)ds}\phi(X_T)\vert X_t=x \right]
\end{align}
To construct the solution to the valuation equation (\ref{PreDefaultEquation}), we define
\begin{align}\label{Iteration}
V^k(t,x) & = \mathbb{E}\left[\int^T_t(c(s,X_s)+\lambda(s,X_s)f(s,X_s,V^{k-1}(s,X_s)))e^{-\int^s_tr(u,X_u) +\lambda(u,X_u)du}ds \right. \nonumber\\ & \hspace*{0.4in} \left. +e^{-\int^T_tr(s,X_s)+\lambda(s,X_s)ds}\phi(X_T)\vert X_t=x\right],
\end{align}
where $V^0(t,x) = U(t,x).$\par
Similar to the proof of Proposition \ref{RiskFreeCVA}, it is easy to see that $V^k, k\ge 0,$ is well-defined and has polynomial growth in $x,$ uniformly in $t.$ The following proposition asserts that $\{V^k\}_{ k\ge 0}$ is a monotonically decreasing sequence of functions.\par
\begin{proposition}\label{PropMonotonicity}$V^k(t,x)\le V^{k-1}(t,x), \forall k\ge 1, (t,x)\in [0,T]\times\mathbb{R}^m$.
\end{proposition}
\begin{proof}
It is easy to see that $V^1(t,x)=V_0(t,x),$ then it follows from Proposition \ref{RiskFreeCVA} that $V^1(t,x)\le V^0(t,x)$. Suppose that $V^{m+1}(t,x)\le V^{m}(t,x).$ Consider $w(t,x):=V^{m+2}(t,x)-V^{m+1}(t,x)$, then we have that
\begin{align*}
w(t,x)= \mathbb{E}\left[\int^T_t\lambda(s,X_s)(f(s,X_s,V^{m+1}(s,X_s))-f(s,X_s,V^{m}(s,X_s)))e^{-\int^s_tr(u,X_u) +\lambda(u,X_u)du}ds\vert X_t=x\right].
\end{align*}
Then $w(t,x)\le 0$ follows from the monotonicity of $f(t,x,y)$ in $y,$ which therefore shows that $V^k(t,x)$ is decreasing with respect to $k$ by mathematical induction.
\end{proof}\par
The following theorem demonstrates the solvability of the valuation equation (\ref{PreDefaultEquation}).
\begin{theorem}\label{Main}
With the notations above, there exists a unique $V\in C^{1,2}([0,T)\times\mathbb{R}^m)\cap C([0,T]\times\mathbb{R}^m)$ with polynomial growth, such that the valuation equation (\ref{PreDefaultEquation}) holds. Moreover, $V(t,x) = \lim_{k\to\infty}V^k(t,x), \forall (t,x)\in [0,T]\times\mathbb{R}^m.$
\end{theorem}
\begin{proof}
It follows from Lemma \ref{LemmaGrowth} (see Appendix) and Proposition \ref{PropMonotonicity} that there exist constants $C>0, n>1$, independent of $k$, such that
\begin{align}\label{EstimatesVk}
\vert V^k(t,x)\vert \le C(\vert x \vert^n+1) , \forall (t,x,k)\in\mathbb{R}_+\times\mathbb{R}^m\times\mathbb{N}.
\end{align}
Let $V(t,x) := \lim_{k\to\infty}V^k(t,x), \forall (t,x)\in [0,T]\times\mathbb{R}^m.$ Then the dominated convergence theorem and the regularity of $f$ yield that $V$ solves the valuation equation (\ref{PreDefaultEquation}) and has polynomial growth.\par
Moreover, as $c, \phi, r$ are locally H{\"o}lder continuous, $\{a_{ij}\}_{i,j\in\{1,2,\ldots,m\}}, b$ are Lipschitz continuous and $\mathcal{L}$ is uniformly parabolic, then it follows from standard arguments in PDE \cite[e.g.][]{aronson1967parabolic} and the Feynman-Kac formula that $V^k\in C^{1,2}([0,T)\times\mathbb{R}^m)\cap C([0,T]\times\mathbb{R}^m)$ and solves the following Cauchy problem
\begin{align*}
\mathcal{L}V^k(t,x)-(r(t,x)+\lambda(t,x))V^k(t,x) + \lambda(t,x)f(t,x,V^{k-1}(t,x))+c(t,x)&=0, \quad (t,x)\in[0,T)\times\mathbb{R}^m,\\V^k(T,x) &= \phi(x), \quad x\in\mathbb{R}^m.
\end{align*}
Let $E:=(0,T)\times (-M,M)^m, M>0$.  According to the interior $L^p$-estimates \citep[e.g.][chapter 7]{lieberman1996second} and the estimate (\ref{EstimatesVk}), we have that, for any $p>1$, there exists a constant $K$, independent of $k$, such that
\begin{align*}
\vert\vert D^2V^k\vert\vert_{p,E}+\left\vert\left\vert \frac{\partial V^k}{\partial t}\right\vert\right\vert_{p,E}\le K.
\end{align*}
Then it follows from the embedding theorem \citep[e.g.][chapter 7]{gilbarg2015elliptic} and the Arzela-Ascoli theorem that  $V\in C^{0,1}([0,T)\times\mathbb{R})^m$.\par
Consider the following standard linear PDE problem,
\begin{align*}
\mathcal{L}H(t,x)-(r(t,x)+\lambda(t,x))H(t,x) + \lambda(t,x)f(t,x,V(t,x))+c(t,x)&=0, \quad (t,x)\in[0,T)\times\mathbb{R}^m,\\H(T,x) &= \phi(x), \quad x\in\mathbb{R}^m.
\end{align*}
By the regularity of the coefficients of $\mathcal{L},$ $\lambda, r, c, f, V$, it is easy to see that the above PDE problem admits a unique classical solution $H$ with polynomial growth. Then the Feynman-Kac formula yields that
\begin{align*}
H(t,x) &= \mathbb{E}\left[\int^T_t(c(s,X_s)+\lambda(s,X_s)f(s,X_s,V(t,X_s)))e^{-\int^s_tr(u,X_u)+\lambda(u,X_u)du}ds\right. \nonumber\\& \hspace*{0.4in} \left. +e^{-\int^T_tr(s,X_s)+\lambda(s,X_s)ds}\phi(X_T)\vert X_t=x \right].
\end{align*}
In other words, we have $
H(t,x) =V(t,x)$
which suggests that $V$ is in $C^{1,2}([0,T)\times\mathbb{R}^m)\cap C([0,T]\times\mathbb{R}^m)$ and has polynomial growth.\par

Finally, we consider the uniqueness of the solution to the valuation equation (\ref{PreDefaultEquation}).
First, using the Feynman-Kac formula on (\ref{PreDefaultEquation}), we have that $V$ solves the following PDE problem
\begin{align}\label{PDECauchy}
\mathcal{L}V(t,x)-(r(t,x)+\lambda(t,x))V(t,x) + \lambda(t,x)f(t,x,V(t,x))+c(t,x)&=0, \quad (t,x)\in[0,T)\times\mathbb{R}^m,\\\label{BoundaryCauchy}V(T,x) &= \phi(x), \quad x\in\mathbb{R}^m.
\end{align}
Second, suppose that $\tilde{V}\in C^{1,2}([0,T)\times\mathbb{R}^m)\cap C([0,T]\times\mathbb{R}^m)$ solves (\ref{PreDefaultEquation}) and has polynomial growth. Let $h := V-\tilde{V},$ then it follows from (\ref{PDECauchy}) and (\ref{BoundaryCauchy}) that
\begin{align*}
\mathcal{L}h(t,x)-(r(t,x)+\lambda(t,x)-e(t,x))h(t,x)&=0, \quad (t,x)\in[0,T)\times\mathbb{R}^m,\\h(T,x) &= 0, \quad x\in\mathbb{R}^m,
\end{align*}
where $e(t,x):= \lambda(t,x)\frac{f(t,x,V(t,x))-f(t,x,\tilde{V}(t,x))}{V(t,x)-\tilde{V}(t,x)}1_{V(t,x)-\tilde{V}(t,x)\neq 0}.$\par
It follows from condition (\ref{MoralHazard}) that $0\le e(t,x)\le \lambda(t,x)$. Therefore, the uniqueness of the solution to the valuation equation (\ref{PreDefaultEquation}) is a result of the maximum principle \citep[eee, e.g.,][chapter 2]{friedman2008partial}.\par
This completes the proof.
\end{proof}\par
In related literature, \cite{kim2016pricing}, \cite{bichuch2018arbitrage} and \cite{brigo2019nonlinear} have focused on the unique solvability of the nonlinear valuation system. These literature usually assumes the boundedness of the hazard rates to guarantee the applicability of the contraction mapping, and thus excluding almost all commonly used stochastic intensity-based models. In contrast, our methodology does not require the boundedness assumption, and  hence allowing the analysis of valuation models with stochastic hazard rates.
\subsection{Risk-free versus replacement closeout: a consequence for valuation }
Analogously, we define the CVA calculated using the replacement closeout by $\Pi(t,x),$ i.e.,
\begin{align*}
\Pi(t,x)= U(t,x)-V(t,x).
\end{align*}
The following theorem verifies that the CVA based on risk-free closeout convention underestimates the corresponding value based on the replacement closeout.
\begin{theorem}\label{ThmCVA}
$\Pi_0(t,x)\le \Pi(t,x), \forall (t,x)\in [0,T]\times\mathbb{R}^m$.
\end{theorem}
\begin{proof}
The result is an immediate consequence of Proposition \ref{PropMonotonicity}  and Theorem \ref{Main}.
\end{proof}\par
Theorem \ref{ThmCVA} points out that neglecting the counterparty's creditworthiness in the recovery calculation is one of the critical sources of the underestimation of CVA. To provide additional insight on the severity of the underestimation, we proceed by valuing a defaultable European put option. In this example, we assume that the dynamics of the underlying asset is given by a one-dimensional geometric Brownian motion with $\mu(t,x) =rx,$ and $ \sigma(t,x) = \sigma x$ in (\ref{SDE}), where $r, \sigma$ are positive constants. We also set $\phi(x) = (K-x)^+, c(t,x)=0$ in (\ref{RiskFreeValue}), (\ref{PreDefaultValueRiskFree}), and (\ref{PreDefaultEquation}), where $K$ is a positive constant. Suppose further that $f(t,x,y) = Ry^+-y^-$ and $\lambda$ is a constant in (\ref{PreDefaultValueRiskFree}), and (\ref{PreDefaultEquation}), with $R$ being a positive constant ranging over $[0,1)$. Finally, we set the CVA calculated with the replacement closeout as an exact value and define the relative error attributed to the use of risk-free closeout by $e(t,x)$, i.e., $e(t,x) =\frac{\Pi(t,x)-\Pi_0(t,x)}{\Pi(t,x)}.$\par
Figure \ref{errorCVA}, which plots the relative error against the counterparty's hazard rate, illustrates that the underestimation is intensified with the worsening of the credit quality of the counterparty since higher hazard rate implies lower credit quality. For example, when  the counterparty's creditworthiness  is as low as that corresponds to 0.3 hazard rate, the degree of underestimation is as high as almost 40\%.  By and large, Theorem \ref{ThmCVA} and the numerical result depicted by Figure \ref{errorCVA} caution against the risk-free approximation in the recovery calculation and thus highlighting the importance of keeping track of the MtM value in the CVA calculation.

\begin{figure}[H]
\label{errorCVA}
\begin{center}
\includegraphics[width=1.0\textwidth]{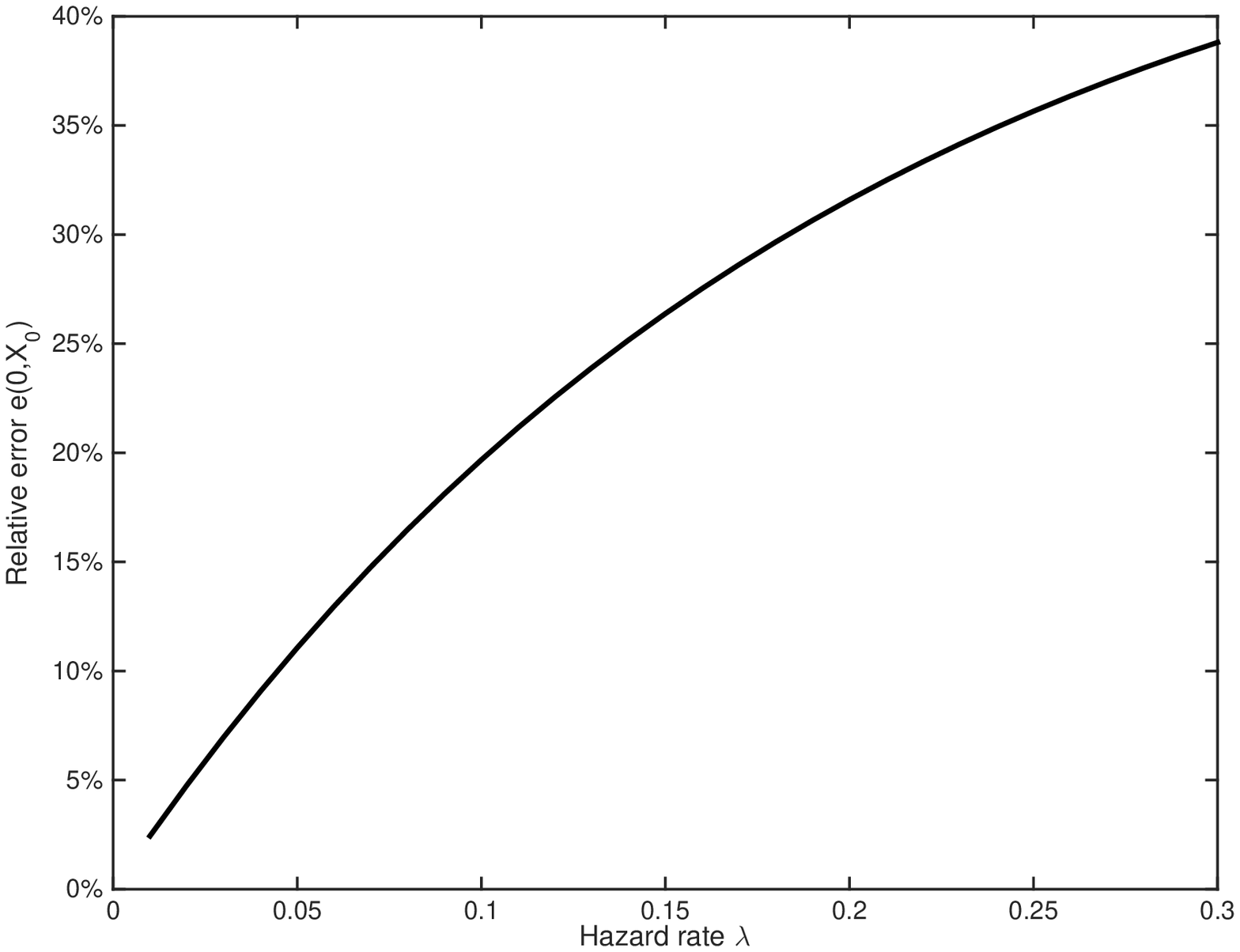}
\begin{minipage}{0.9\textwidth}
{\footnotesize \sf \singlespacing Fig. \ref{errorCVA}. The magnitude of the underestimate of CVA. \textit{Notes.} The figure plots the variation of relative error resulting from the neglect of the counterparty's creditworthiness with respect to the degree of hazard rate $\lambda$. The graph uses the parameters $r = 0.05, R = 0.5, \sigma=0.2, K=X_0=1, T=10.$
}
\end{minipage}
\end{center}
\end{figure}
\subsection{An extension: the bilateral case}
Our analysis so far has focused on the unilateral CVA for which only the counterparty's default is considered, we extend the theoretical results to the bilateral case in this subsection\footnote{See, for example, \cite{gregory2009being} and \cite{brigo2014arbitrage}, for the valuation of bilateral counterparty credit risk under the risk-free closeout convention.}.\par
In the bilateral case, in addition to the risk due to the potential default of the counterparty, the possibility of the investor's default is also considered. Let $\tau'$ be a non-negative random variable on $(\Omega,\mathcal{F},\mathbb{P})$ that represent the default time of the investor. Define the associated default process by $H'_t,$ i.e., $H'_t=1_{\tau'\le t}$. At time $t$, the conditional survival probability of the investor with respect to the accessible information $\mathcal{F}_t$ is given by $\mathbb{P}(\tau'>t\vert\mathcal{F}_t):=e^{-\int^t_0\bar{\lambda}(s,X_s)ds}$, where $\bar{\lambda}: \mathbb{R}_+\times\mathbb{R}^m\rightarrow\mathbb{R}_+,$ satisfies the same growth and regularity conditions as $\lambda.$ Following the standard literature on the hazard rate models \citep[e.g.,][ chapter 9]{bielecki2013credit}, we suppose that $\tau$ and $\tau'$ are conditionally independent with respect to $\{\mathcal{F}_t\}_{t\ge 0}$, i.e., $\mathbb{P}(\tau>t,\tau'>t\vert\mathcal{F}_t) = \mathbb{P}(\tau>t\vert\mathcal{F}_t) \mathbb{P}(\tau'>t\vert\mathcal{F}_t)=e^{-\int^t_0\lambda(t,X_s)+\bar{\lambda}(s,X_s)ds}.$\par
Denote the closeout function for the investor by $\bar{f}(t,x,y): \mathbb{R}_+\times\mathbb{R}^m\times\mathbb{R}\rightarrow\mathbb{R}$, where $\bar{f}$ is locally H{\"o}lder continuous in $x$, uniformly in $t,y$. Suppose that $\bar{f}(t,x,y)\ge y, \forall (t,x,y)\in \mathbb{R}_+\times\mathbb{R}^m\times\mathbb{R}$. In order to rule out the moral hazard, we similarly assume that $\bar{f}$ satisfies inequality (\ref{MoralHazard}). A specific form of $\bar{f}$ is $y^+-R'y^-, R'\in[0,1)$, which is used as a closeout function in the calculation of DVA. Also, it is  immediate to see that our closeout functions are sufficiently general to accommodate different specifications of cash flows that arises in the studies of FVA and XVA \citep[see, e.g.,][]{kim2016pricing,brigo2019nonlinear}.\par
Let us consider the following discounted payoff process of a defaultable claim
\begin{align}\label{CashFlowBilateral}
\Upsilon_t &=\int^{\tau\wedge\tau'\wedge T}_te^{-\int^s_t r(u,X_u)du}c(s,X_s)ds + 1_{t<\tau\le\tau'\le T}\chi_1(\tau,X_{\tau})e^{-\int^{\tau}_t r(s,X_s)ds}\nonumber\\& +1_{t<\tau'<\tau\le T}\chi_2(\tau',X_{\tau'})e^{-\int^{\tau'}_t r(s,X_s)ds}+ \phi(X_T)e^{-\int^T_t r(s,X_s)ds} 1_{\tau\wedge\tau'>T},
\end{align}
where $\chi_1$ ($\chi_2$) denotes the recovery when the counterparty (investor) defaults.\par
Let the pre-default value of the financial claim with replacement closeout be defined by  $\Psi: \mathbb{R}_+\times\mathbb{R}^m\rightarrow\mathbb{R}$ and let $\chi_1(\tau,X_{\tau})=f(\tau,X_{\tau},\Psi(\tau,X_{\tau}))$ and $\chi_2(\tau',X_{\tau'})=\bar{f}(\tau',X_{\tau'},\Psi(\tau',X_{\tau'}))$ in (\ref{CashFlowBilateral}). Then it follows from the standard argument \citep[e.g.][chapter 9]{bielecki2013credit} on credit risk that the pre-default value $\Psi$ is given by 
\begin{align}\label{PreDefaultEquationBilateral}
\Psi(t,x) &= \mathbb{E}\left[\int^T_t(c(s,X_s)+\lambda(s,X_s)f(s,X_s,\Psi(t,X_s))\right. \nonumber\\
&+\bar{\lambda}(s,X_s)\bar{f}(s,X_s,\Psi(t,X_s)))e^{-\int^s_tr(u,X_u)+\lambda(u,X_u)+\bar{\lambda}(u,X_u)du}ds\nonumber\\
& \left. +e^{-\int^T_tr(s,X_s)+\lambda(s,X_s)+\bar{\lambda}(s,X_s)ds}\phi(X_T)\vert X_t=x\right].
\end{align}
We close this subsection by studying the solvability of the valuation equation (\ref{PreDefaultEquationBilateral}) and the impact of the replacement closeout on the valuation. To this end, we set a benchmark as follows:
\begin{align}\label{BenchmarkBilateral}
\Psi_0(t,x) &= \mathbb{E}\left[\int^T_t(c(s,X_s)+\lambda(s,X_s)f(s,X_s,V(t,X_s))\right.\nonumber\\
&+\bar{\lambda}(s,X_s)\bar{f}(s,X_s,V(t,X_s)))e^{-\int^s_tr(u,X_u)+\lambda(u,X_u)+\bar{\lambda}(u,X_u)du}ds\nonumber\\
&\left. +e^{-\int^T_tr(s,X_s)+\lambda(s,X_s)+\bar{\lambda}(u,X_u)ds}\phi(X_T)\vert X_t=x\right].
\end{align}\par
\begin{theorem}\label{MainBilateral}
With the notations above, there exists a unique $\Psi\in C^{1,2}([0,T)\times\mathbb{R}^m)\cap C([0,T]\times\mathbb{R}^m)$ with polynomial growth, such that the valuation equation (\ref{PreDefaultEquationBilateral}) holds. Moreover, $\Psi(t,x)\ge \Psi_0(t,x), \forall (t,x)\in [0,T]\times\mathbb{R}^m.$
\end{theorem}
\begin{proof}
Using Lemmas \ref{Lem:MonotonicityStep1Bilateral}, \ref{Lem:MonotonicityBilateral} and \ref{Lem:UpperBoundBilateral}, the proof is essentially the same as the proofs of Theorems \ref{Main} and \ref{ThmCVA}. Thus we omit it.
\end{proof}\par
Theorem \ref{MainBilateral}, which extends Theorem \ref{Main} to the bilateral case, demonstrates the unique solvability of the valuation equation (\ref{PreDefaultEquationBilateral}). Also, comparing (\ref{PreDefaultEquationBilateral}) and  (\ref{BenchmarkBilateral}), one can see that the pricing formula (\ref{BenchmarkBilateral}) ignores the investor's creditworthiness in the recovery calculation. Hence, Theorem \ref{MainBilateral} shows that the neglecting the investor's creditworthiness in the recovery calculation would result in overestimating the counterparty credit risk.
\section{Algorithms}\label{Sec:Algorithm}
In this section, we introduce a deep learning-based algorithm for computing CVA. We work within the Deep BSDE framework established by \cite{weinan2017deep}.  While \cite{weinan2017deep} treat the gradient of the value function as a sequence of state functions at different time nodes, and thus use a sequence of different neural networks to approximate the gradient, the algorithm  proposed in this paper, on the other hand, considers the gradient to be a function of time and states and thus allows us to use one unified neutral network to approximate the gradient.
\subsection{The BSDE formulation}
As CVA is defined by the difference between the risk-free value and the pre-default value of a financial claim, and the risk-free value is a special case of the pre-default value with $\lambda(t,X_t)\equiv 0$, it suffices to focus on the valuation equation (\ref{PreDefaultEquation}) of the pre-default value of the financial claim.
An immediate result of Theorem \ref{Main} and Feynman-Kac formula is that $V$, defined by the valuation equation (\ref{PreDefaultEquation}), is the classical solution to the following PDE problem:
\begin{align}\label{PDECauchy1}
\mathcal{L}V(t,x)-(r(t,x)+\lambda(t,x))V(t,x) + \lambda(t,x)f(t,x,V(t,x))+c(t,x)&=0, \quad (t,x)\in[0,T)\times\mathbb{R}^m,\\\label{BoundaryCauchy2}V(T,x) &= \phi(x), \quad x\in\mathbb{R}^m.
\end{align}
It follows from  \cite{el1997backward} that the nonlinear PDE problem, which comprises (\ref{PDECauchy1}) and (\ref{BoundaryCauchy2}), is related to a BSDE in the sense that\footnote{See, e.g., \cite{crepey2015bilateral}, \cite{crepey2015bilateralb} and  \cite{bichuch2018arbitrage}, for more discussion about the BSDE formulation for the valuation of credit risk. }
\begin{align}\label{BSDE}
V(t,X_t) = \phi(X_T)+\int^T_tF(s,X_s,V_s)ds-\int^T_t(\sigma(s,X_s)^T\nabla_x V(s,X_s))^TdW_s.
\end{align}
Here $\{X_t\}_{t\in[0,T]}$ is given by the SDE (\ref{SDE}) and $F(t,x,y)=c(t,x)+\lambda(t,x)f(t,x,y)-(r(t,x)+\lambda(t,x))y$.\par
\subsection{The stochastic control problem }
We apply a time discretization to (\ref{SDE}) and (\ref{BSDE}). More specifically, let $0=t_0<t_1<t_2<...<t_N=T$ and denote $\Delta W_i = W_{t_{i+1}}-W_{t_i}, \Delta t_i=t_{i+1}-t_i, i=0,1,2,...,N-1.$ Then applying the Euler-Maruyama scheme results in
\begin{align}\label{DiscretizedSDE}
 X_{t_{i+1}} \approx X_{t_i} + \mu(t_i, X_{t_i}) \Delta t_i + \sigma(t_i, X_{t_i})\Delta W_i, \quad X_0=x
\end{align}
\begin{align}\label{DiscretizedBSDE}
V_{t_{i+1}}  &\approx V_{t_i} - F(t_i, X_{t_i}, V_{t_i}) \Delta t_i+ \nabla_x V(t_i, X_{t_i})^T \sigma(t_i, X_{t_i})\Delta W_i, \quad V_0 = v.
\end{align}
Taking $\nabla V$ and $V_0$ as decision variables, we then obtain the stochastic control problem
\begin{align}\label{Pr:Control}
\min_{v\in\mathbb{R}^m,\{z_i\}_{i=0,1,2,...,N-1}, z_i\in\mathbb{R}^m}\mathbb{E}\vert V_{t_N}-\phi(X_{t_N})\vert^2,
\end{align}
such that
\begin{align}\label{BSDE:Control}
V_{t_{i+1}}  &= V_{t_i} - F(t_i, X_{t_i}, V_{t_i}) \Delta t_i+ z_i^T \sigma(t_i, X_{t_i})\Delta W_i, \quad V_0 = v,
\end{align}
where $X_{t_i}, i=0,1,2,...,N-1$ is given by (\ref{DiscretizedSDE}).\par
In order to solve the stochastic control problem (\ref{Pr:Control}), we use Monte Carlo method to simulate $L$ paths $(X^l_{t_i},V^l_{t_i})_{i=0,1,2,...,N-1}, l=1,2,3...,L,$ using (\ref{DiscretizedSDE}) and (\ref{BSDE:Control}). Then the control problem (\ref{Pr:Control}) reads
\begin{align}\label{ObjectiveFunctional}
\min_{v\in\mathbb{R}^m,\{z_i\}_{i=0,1,2,...,N-1}, z_i\in\mathbb{R}^m}\frac{1}{L}\sum_{l=1}^L( V^l_{t_N}-\phi(X^l_{t_N}))^2,
\end{align}
where $(X^l_{t_N},V^l_{t_N}), l=1,2,3...,L,$ are the simulated  sample paths correspond to (\ref{DiscretizedSDE}) and  (\ref{BSDE:Control}).
\subsection{The single network-based algorithm }
Motivated by \cite{raissi2018forward}, we propose a single network-based algorithms for solving the stochastic control problem (\ref{Pr:Control}). We consider $\nabla_x V$ to be a function of $(t,x)$ and approximate it by a network $\mathcal{N}$, i.e.,
\begin{align}\label{Network}
\nabla_xV(t_i,X_{t_i})\thickapprox \mathcal{N}(t_i,X_{t_i}\vert \theta), 0\le i\le N-1,
\end{align}
where $\theta=\{\theta_j\}_{j=1,2,3...,J}$ denotes the trainable parameters in the network.\par
Equipping the Discretized BSDE (\ref{DiscretizedBSDE}) with the neural network, we have 
\begin{align}\label{NetworkBSDE}
\mathcal{V}^{\theta}_{t_{i+1}}  = \mathcal{V}^{\theta}_{t_i} - F(t_i, X_{t_i}, \mathcal{V}^{\theta}_{t_i}) \Delta t_i+  \mathcal{N}(t_i,X_{t_i})^T \sigma(t_i, X_{t_i})\Delta W_i, \quad \mathcal{V}^{\theta}_0 = v.
\end{align}
Replacing the objective function (\ref{Pr:Control}), we define the loss function used in the Deep BSDE method by $\iota(y,\theta),$ i.e.,
\begin{align}
\iota(v,\theta) = \mathbb{E}\vert \mathcal{V}^{\theta}_{t_N}-\phi(X_{t_N})\vert^2.
\end{align}
We then obtain the control problem considered in the Deep BSDE method
\begin{align}\label{Pr:ControNetwork}
\min_{v\in\mathbb{R}^m,\theta\in\mathbb{R}^J }\iota(v,\theta),
\end{align}
subject to (\ref{DiscretizedSDE}) and (\ref{NetworkBSDE}).\par
Having obtained the control problem (\ref{Pr:ControNetwork}), we are poised to use the a stochastic gradient descent-type (SGD) algorithm to seek the optimal solution. The algorithm is summarized as follows.

\begin{algorithm}[H]
\caption{The single network-based Deep BSDE solver}\label{alg iteration}
\begin{algorithmic}[1]
\Function{DBSDE}{$N$, $L$, SDE, BSDE}: \Comment{$N$ time steps, $L$ paths for Monte Carlo loop, SDE defined by (\ref{DiscretizedSDE})  and BSDE defined by (\ref{NetworkBSDE}).}
\State Fix the neural network architecture $\mathcal{N}$.
\State Initialize trainable parameters $(y,\theta)$.
\For{$i = 1$ to $n$}: \Comment{$n$ iterations for training the neural network.}
\State Simulate $L$ paths of $X$ and $\mathcal{V}$ by SDE and BSDE respectively.
\State Replace $\iota(v,\theta)$ by $\frac{1}{L}\sum_{l=1}^L( \mathcal{V}^{\theta,l}_{t_N}-\phi(X^l_{t_N}))^2.$ \Comment{ $(X^l,\mathcal{V}^{\theta,l})_{1\le l\le L}$ are the sample paths of  $X$ and $\mathcal{V^{\theta}}$.}
\State Compute the loss
\begin{align*}
    \frac{1}{L}\sum_{l=1}^L( \mathcal{V}^{\theta,l}_{t_N}-\phi(X^l_{t_N}))^2,
\end{align*}
\State Reduce the loss by an SGD algorithm and update $(y_i,\theta_i)$.
\EndFor
\State \Return updated parameters $(y^*,\theta^*).$ \Comment{ $(y^*,\theta^*)$ are all parameters of the trained network $\mathcal{N}^*$. Thus, the algorithm returns the trained network $\mathcal{N}^*$.}
\EndFunction
\end{algorithmic}
\end{algorithm}
In practice, in order to increase the reliability of the algorithm, one usually runs independently the algorithm many times with different initial values and then takes the average as the value function. Using this idea, we obtain the algorithms for the pre-default value of a defaultable claim with the replacement closeout and corresponding CVA as follows.
\begin{algorithm}[H]
\caption{The algorithm for the pre-default value of a defaultable claim with the replacement closeout}\label{alg replacement}
\begin{algorithmic}[1]
    \Function{VALUE\_REPLACEMENT}{$N$, $L$, SDE, BSDE, $M$}:  \Comment{$N$ time steps, $L$ paths for Monte Carlo loop, SDE defined by (\ref{DiscretizedSDE}),  BSDE defined by (\ref{NetworkBSDE}) and $M$ independent trials.}
    \For{$i=1$ to $M$}:
    \State Import the function from Algorithm \ref{alg iteration} and set up the initial value of trainable parameters $(y_i,\theta_i)$.
    \State $(y_i^*,\theta_i^*)\leftarrow$ DBSDE($N$, $L$, SDE, BSDE).
    \EndFor
    \State \Return $V(t_0, X_{t_0}) \leftarrow \frac{1}{M}\sum_{i=1}^My^*_i$. \Comment{ $V$ is the pre-default value of the defaultable claim with the replacement closeout.}
    \EndFunction
\end{algorithmic}
\end{algorithm}

\begin{algorithm}[H]
\caption{The CVA solver}\label{alg CVA}
\begin{algorithmic}[1]
    \Function{CVA}{$N$, $L$, SDE, BSDE, BSDE\_U, $M$}:\Comment{$N$ time steps, $L$ paths for Monte Carlo loop, SDE defined by (\ref{DiscretizedSDE}), BSDE defined by (\ref{NetworkBSDE}), BSDE\_U defined by (\ref{NetworkBSDE}) with $\lambda(t,x)\equiv 0$ and $M$ independent trials.}

    \State Import the function from Algorithm \ref{alg replacement}.
    \State $y_{0}^*\leftarrow$ VALUE\_REPLACEMENT($N$, $L$, SDE, BSDE, M).\Comment{$y_0$ is the pre-default value of the defaultable claim with the replacement closeout at $(t_0,X_{t_0})$.}
       \State $y_{1}^*\leftarrow$ VALUE\_REPLACEMENT($N$, $L$, SDE, BSDE\_U, $M$).\Comment{$y_1$ is the value of the risk-free counterpart of the defaultable claim at $(t_0,X_{t_0})$.}

    \State \Return $c\leftarrow$ $y_{1}^*-y_0^*$.\Comment{$c$ is the CVA for the defaultable claim with the replacement closeout. }
    \EndFunction
\end{algorithmic}
\end{algorithm}
\section{Numerical Results and Discussion}\label{sec:NRD}
In this section, we test our algorithm on a defaultable European option by assuming $\phi(x) = (dK-\sum_{i=1}^dx_i)^+, \lambda(t,x) = \lambda, c(t,x) = 0, r(t,x)=r, f(t,x,y) = Ry^+-y^-,$ where $K, \lambda, R$ are positive constants. Moreover, we suppose the underlying asset follows a multi-dimensional Geometric Brownian motion,
\begin{align*}
dX_{it} = \mu_iX_{it}dt+\sigma_{i}X_{it}dW_{it}, i = 1,2,\dots,d,
\end{align*}
where $\mu_i, \sigma_i$ are constants and $W= (W_1,W_2,\ldots,W_d)^T$  is a $d$-dimensional Brownian motion.\par
The single network-based algorithm is tested with the stacked LSTM architecture with $3$ hidden layers (all $d+10$ dimensions). We choose the hyperbolic tangent function as the activation function and use the Adam optimizer \citep{kingma2014adam} for the SGD process. We implement our algorithm on the defaultable option with $N=100$ (number of time steps), $L=64$ (number of paths for the Monte Carlo loop). All tests in this section are run on a PC with a 2.80GHz Intel Core i5-8400 CPU and an NVIDIA GeForce RTX 2060 GPU using Google TensorFlow in Python.
\subsection{CVA}
The nonlinearity resulting from the replacement closeout is regarded as a major obstacle to the application, especially when the dimension (the number of the risk factors) of the valuation model is high. Because of the power of the deep learning technique, we are now able to calculate the CVA with replacement closeout in high dimensional contexts. Table \ref{TableCVA} lists the CVA for the defaultable European option over different values of dimensional case $d$, where $d = 5,10,20,50,100.$
\begin{table}[H] \label{TableCVA}
               \begin{center}
           \begin{tabular}{|c|c|c|c|c|c|}
\hline
Dimension   & 5    & 10   & 20   & 50   &100 \\ \hline
CVA & $0.04487$  & $0.0876$  & $0.1794$  & $0.4495$  & $0.8288$ \\ \hline
\end{tabular}
       \vspace{0.5cm}
\begin{minipage}{0.9\textwidth}
{\footnotesize \sf \singlespacing Table \ref{TableCVA}. The CVA for the defaultable European put option. \textit{Notes.} The table compares the CVA for the defaultable European put option, calculated with the single LSTM DBSDE solver in the present paper. The table uses the parameters $r = 0.03, \mu=0.05, \sigma=0.2, \lambda=0.1, R = 0.4, K=1, X_{i0}=0.8, i=1,2,\ldots,d, T=1$. The number of the independent trials is $M=5$.
}
\end{minipage}
\end {center}
\end{table}

\subsection{Comparison with the multi-FC DBSDE solver (\citeauthor{weinan2017deep} \citeyear{weinan2017deep})}
In this subsection, we show the accuracy of our algorithm. To this end, we compare our algorithm with the multi-FC DBSDE solver \citep{weinan2017deep}, whose validation for the valuation of defaultable claims has been examined by \cite{gnoatto2020deep}. While the results obtained from multi-FC DBSDE serve as a benchmark, Table~\ref{ComparisonValues} demonstrates the ability of the single network-based algorithm in providing accurate numerical approximations even for dimension as high as 100. \par
Figure \ref{Convergence} compares the algorithm in the present paper with the mult-FC DBSDE solver in terms of the iteration number required for the convergence in the $100$-dimensional case. The left panel shows that the single LSTM DBSDE solver requires fewer iterations to reach the convergence than the mult-FC DBSDE solver. The right panel analyzes the approximations in the convergence state. The shaded areas depict the means and the ranges of the 20 independent runs of the two algorithms. This panel clearly demonstrates that the proposed single network-based algorithm is more stable with the random choice of the initial parameters than the mult-FC DBSDE solver.\par
\begin{table}[H] \label{ComparisonValues}
               \begin{center}
               \begin{tabular} {|c|c|c|c|c|c|c|c|c|c|}
                    \hline Dimension  & 5 & 10 & 20 & 50 & 100   \\
                    \hline Multi-FC  & $0.7378$ & $1.4578$ & $2.9082$ & $7.2663$ & $14.5331$  \\
                    \hline Single-LSTM  & $0.7285$ & $1.4548$ & $2.9013$ & $7.2690$ & $14.5273$  \\
                    \hline Relative Error  & $1.26\%$ & $0.21\%$ & $0.24\%$ & $0.04\%$ & $0.04\%$  \\
                    \hline
               \end{tabular}
       \vspace{0.5cm}
\begin{minipage}{0.9\textwidth}
{\footnotesize \sf \singlespacing Table \ref{ComparisonValues}. The pre-default values of the defaultable European put option. \textit{Notes.} The table compares the pre-default values of the defaultable European put option, calculated with the multi-FC DBSDE solver and the single LSTM DBSDE solver in the present paper. The table uses the parameters $r = 0.03, \mu=0.05, \sigma=0.2, \lambda=0.1, R = 0.4, K=1, X_{i0}=0.8, i=1,2,\ldots,d, T=1$. The number of the independent trials is $M=5$. We set the results obtained from the Multi-FC DBSDE solver as a benchmark. Thus the relative error is given by $\vert\frac{V^s-V^m}{V^m}\vert,$ where $V^m$ and $V^s$ are the pre-default values calculated with the multi-FC DBSDE solver and the single LSTM DBSDE solver respectively.
}
\end{minipage}
\end {center}
\end{table}
\begin{figure}[H]
\label{Convergence}
\begin{center}
\includegraphics[width=1.0\textwidth]{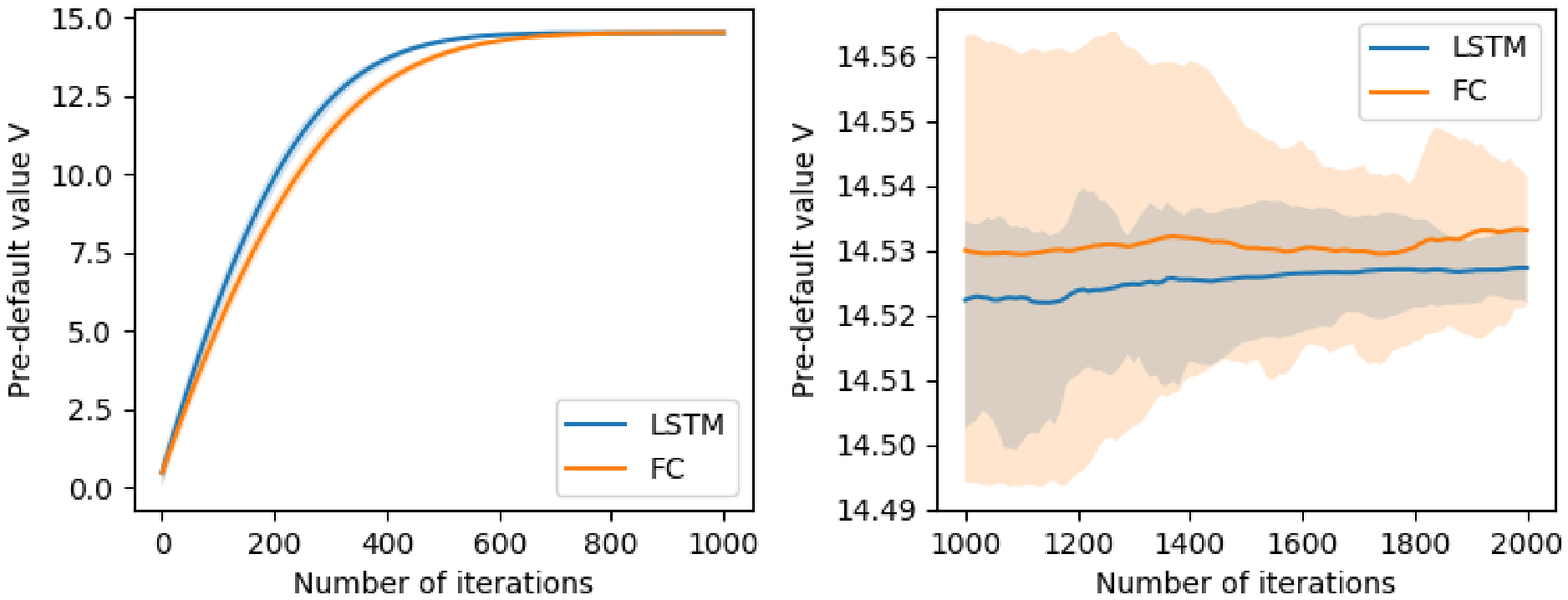}
\begin{minipage}{0.9\textwidth}
{\footnotesize \sf \singlespacing Fig. \ref{Convergence}. Approximation of the pre-default value of the option against the number of iterations. \textit{Notes.} The left panel compares the iteration numbers required for the convergence of the single LSTM DBSDE solver and the multi-FC DBSDE solver respectively. The right panel breaks down the convergence state and demonstrates the stability of the two solvers. The shade areas depict the ranges of the $20$ independent runs of the algorithms with different initial trainable parameters. The figure uses the parameters $r = 0.03, \mu=0.05, \sigma=0.2, \lambda=0.1, R = 0.4, K=1, X_{i0}=0.8, i=1,2,\ldots,d, T=1$.
}
\end{minipage}
\end{center}
\end{figure}

\subsection{Risk-free versus replacement closeout: a consequence for computational cost}
Recall that in the proof of Theorem~\ref{Main},  the pre-default value of a defaultable claim with the replacement closeout is obtained with an iterative process, while the pre-default value with the risk-free closeout is the result of the second iteration. Thus, it seems to suggest that the valuation with replacement closeout is computationally more burdensome. However, the Deep BSDE framework gives a different pattern on the computational cost. A distinctive feature of Deep BSDE solvers is that this type of algorithm tackles the nonlinear problems directly, rather than linearizing the problem using a sequence of iterations. As a result, using deep learning to  compute the replacement closeout is computationally more efficient than computing the risk-free closeout. This is validated in Table~\ref{ComparisonValuesCloseout} which shows that the computational time for the valuation of the defaultable option with the  replacement closeouts is less than half of the corresponding risk-free closeouts. The saving in computational time holds consistently over the dimensions we have compared.
\begin{table}[H]
\label{ComparisonValuesCloseout}
               \begin{center}
               \begin{tabular}{|l|l|l|l|l|l|}
\hline
Dimension   & 5    & 10   & 20   & 50   &100 \\ \hline
Replacement closeout & 318s  & 352s  & 405s  & 599s  & 857s \\ \hline
Risk-free closeout    & 668s  & 751s  & 873s  & 1299s & 1879s \\ \hline
\end{tabular}
       \vspace{0.5cm}
\begin{minipage}{0.9\textwidth}
{\footnotesize \sf \singlespacing Table \ref{ComparisonValuesCloseout}. The running time for the computation of the pre-default values of the defaultable option with the risk-free and replacement closeout conventions respectively. \textit{Notes.}
The table compares the risk-free and replacement closeout conventions in terms of computational cost. All the results are obtained with the single LSTM DBSDE solver. The table uses the parameters $r = 0.03, \mu=0.05, \sigma=0.2, \lambda=0.1, R = 0.4, K=1, X_0=0.8, T=1$. The running time is the average running time of $5$ independent trials.
}
\end{minipage}
\end {center}
\end{table}

\section{Conclusions}\label{Sec:Conclusions}
In this paper, we have developed a general framework to underpin the replacement closeout in the valuation of defaultable claims and counterparty credit risk.
The majority of the paper is focused on how to handle the nonlinear valuation system attributed to the replacement closeout.
In the theoretical part, we show the unique solvability of the nonlinear system and thus justifying the valuation model from a mathematical perspective. Moreover, we study the impact of the replacement on calculation of CVA. Compared to the traditional risk-free closeout, the replacement closeout captures the creditworthiness of the counterparty. Our result shows that the replacement closeout yields a higher CVA than the conventional risk-free counterpart and thus cautioning that lack of the incorporation of the counterparty's credit quality in the recovery calculation would result in an underestimate of CVA.\par
The nonlinearity arising from the replacement closeout poses a computing challenge for high dimensional valuation models. To address this issue, we develop a Deep BSDE algorithm for the valuation of counterparty risk with the replacement closeout. In contrast to the existing BSDE algorithm in the credit risk literature, our algorithm is single network based. Numerical tests demonstrate that our algorithm works satisfactory in high dimensional models and would show better convergence stability than the conventional multi-network based algorithm. Finally, we conduct numerical analysis to show the impact of the replacement closeout on the computational cost. Our numerical results show that the computation with replacement closeout would consume less running time than the risk-free closeout for the valuation of defaultable claims, and thus suggesting that the Deep BSDE type algorithm would help clear the obstacles for the application of the replacement closeout.

\appendix
\renewcommand{\thesubsection}{\Alph{section}.\arabic{subsection}}
\section{Appendix: Lemmas}

\begin{lemma}\label{LemmaGrowth1}
Define \begin{align*}
U_0(t,x) = \mathbb{E}\left[\int^T_t-\vert c(s,X_s)\vert e^{-\int^s_tr(u,X_u)du}ds-e^{-\int^T_tr(s,X_s)ds}\vert \phi(X_T)\vert\vert X_t=x\right],
\end{align*}
\begin{align*}
U_k(t,x) &= \mathbb{E}\left[ \int^T_t(-\vert c(s,X_s)\vert+\lambda(s,X_s)f(s,X_s,U_{k-1}(t,X_s)))e^{-\int^s_tr(u,X_u)+\lambda(u,X_u)du}ds\right.\nonumber\\& \hspace*{0.4in} \left. -e^{-\int^T_tr(s,X_s)+\lambda(s,X_s)ds}\vert \phi(X_T)\vert\vert X_t=x\right], k\ge 1.
\end{align*}
\begin{align*}
J(t,x) &=  \mathbb{E}\left[\int^T_t(-\vert c(s,X_s)\vert+\lambda(s,X_s)f(s,X_s,0))e^{-\int^s_tr(u,X_u)du}ds\right.\nonumber\\&\hspace*{0.4in} \left. -e^{-\int^T_tr(s,X_s)ds}\vert \phi(X_T)\vert\vert X_t=x\right]
\end{align*}
Then \begin{align*}
U_k(t,x)\ge J(t,x), \quad\forall (t,x)\in[0,T]\times\mathbb{R}^m,  \forall k\ge 1.
\end{align*}
\end{lemma}
\begin{proof}
It follows from the Feynman-Kac formula and the polynomial growth of $U_k$ (See Proposition \ref{Prop:PolynomialGrowth}.) that
\begin{align}\label{Proof:equation}
\mathcal{L}U_k(t,x)-(r(t,x)+\lambda(t,x))U_k(t,x)-\vert c(t,x)\vert+\lambda(t,x)f(t,x,U_{k-1}(t,x))=0,
\end{align}
Thanks to Proposition \ref{PropMonotonicity}, we have that
\begin{align}\label{Proof:Monotone}
0\ge U_{k-1}(t,x)\ge U_k(t,x), \forall k\ge 1, \forall (t,x)\in[0,T]\times\mathbb{R}^m.
\end{align}
Moreover, it is to see from Condition (\ref{MoralHazard}) that
\begin{align*}
f(x) \ge f(0) + x, \forall x\le 0.
\end{align*}
Thus, from (\ref{Proof:Monotone}) and (\ref{Proof:equation}), we have that
\begin{align*}
0=&\mathcal{L}U_k(t,x)-(r(t,x)+\lambda(t,x))U_k(t,x)-\vert c(t,x)\vert+\lambda(t,x)f(t,x,U_{k-1}(t,x))\\\ge &\mathcal{L}U_k(t,x)-(r(t,x)+\lambda(t,x))U_k(t,x)-\vert c(t,x)\vert+\lambda(t,x)f(t,x,U_{k}(t,x))\\\ge &\mathcal{L}U_k(t,x)-(r(t,x)+\lambda(t,x))U_k(t,x)-\vert c(t,x)\vert+\lambda(t,x)(f(t,x,0)+U_k(t,x)),
\end{align*}
i.e.,
\begin{align}\label{Proof:Inequality1}
\mathcal{L}U_k(t,x)-r(t,x)U_k(t,x)-\vert c(t,x)\vert+\lambda(t,x)f(t,x,0)\le 0.
\end{align}
Using the Feynman-Kac formula on $J(t,x)$, we have that
\begin{align}\label{Proof:Inequality2}
\mathcal{L}J(t,x)-r(t,x)J(t,x)-\vert c(t,x)\vert+\lambda(t,x)f(t,x,0)= 0.
\end{align}
Thanks to (\ref{Proof:Inequality1}) and (\ref{Proof:Inequality2}), we then have
\begin{align}\label{Proof:Inequality3}
\mathcal{L}(J(t,x)-U_k(t,x))-r(t,x)(J(t,x)-U_k(t,x))\ge  0.
\end{align}
Then $J(t,x)\le U_k(t,x)$ follows from the maximum principle and the fact that $J(T,x)-U_k(T,x)=0$.\par
This completes the proof.\par
\end{proof}

\begin{lemma}\label{LemmaGrowth}
\begin{align*}
U(t,x)\ge V^k(t,x)\ge J(t,x), \quad\forall (t,x)\in[0,T]\times\mathbb{R}^m, \forall k\ge 1.
\end{align*}
\end{lemma}
\begin{proof}
It is easy to see that $U_0(t,x)\le U(t,x)$. Then the monotonicity of $f(t,x,y)$ in $y$ yields that $U_1(t,x)\le V^1(t,x)$. Following the mathematical induction, we have that $U_k(t,x)\le V^k(t,x), \forall k\ge 1.$ Then the conclusion follows from Lemma \ref{LemmaGrowth1} and Proposition \ref{PropMonotonicity}.
\end{proof}\par

\begin{lemma}\label{Lem:MonotonicityStep1Bilateral}
$\Psi_0(t,x)\ge V(t,x), \forall k\ge 1, (t,x)\in [0,T]\times\mathbb{R}^m$.
\end{lemma}
\begin{proof}
Similar to the proof of Proposition \ref{RiskFreeCVA}, one can see that $\Psi_0$ has polynomial growth and is the classical solution to the following Cauchy problem,
\begin{align*}
&\mathcal{L}\Psi_0(t,x)-(\lambda(t,x)+\bar{\lambda}(t,x)+r(t,x))\Psi_0(t,x)+\lambda(t,x)f(t,x,V(t,x))\\&+\bar{\lambda}(t,x)\bar{f}(t,x,V(t,x))+c(t,x)=0, \quad (t,x)\in[0,T)\times\mathbb{R}^m,\\&\Psi_0(T,x) = \phi(x), \quad x\in\mathbb{R}^m.
\end{align*}
It follows from Theorem \ref{Main} and the Feynman-Kac formula that $V$ is the classical solution to the following Cauchy Problem
\begin{align*}
\mathcal{L}V(t,x)-(\lambda(t,x)+r(t,x))V(t,x)+\lambda(t,x)f(t,x,V(t,x))+c(t,x)&=0, \quad (t,x)\in[0,T)\times\mathbb{R}^m,\\V(T,x) &= \phi(x), \quad x\in\mathbb{R}^m.
\end{align*}
Let $S(t,x): = \Psi_0(t,x) - V(t,x).$ Then some algebra yields that
\begin{align*}
\mathcal{L}S(t,x)-(\lambda(t,x)+\bar{\lambda}(t,x)+r(t,x))S(t,x)+\bar{\lambda}(t,x)(\bar{f}(t,x,V(t,x))-V(t,x))&=0, \quad (t,x)\in[0,T)\times\mathbb{R}^m,\\S(T,x)& = 0, \quad x\in\mathbb{R}^m.
\end{align*}
As $\bar{f}(t,x,V(t,x))-V(t,x)\ge 0,$ we have that,
\begin{align*}
\mathcal{L}S(t,x)-(\lambda(t,x)+\bar{\lambda}(t,x)+r(t,x))S(t,x)&\le 0, \quad (t,x)\in[0,T)\times\mathbb{R}^m.
\end{align*}
Then it follows from the maximum principle that
\begin{align*}
S(t,x) \ge 0, \quad (t,x)\in[0,T]\times\mathbb{R}^m,
\end{align*}
This completes the proof.
\end{proof}
\begin{lemma}\label{Lem:MonotonicityBilateral}
Let \begin{align*}
\Psi^k(t,x)&= \mathbb{E}\left[\int^T_t(c(s,X_s)+\lambda(s,X_s)f(s,X_s,\Psi^{k-1}(t,X_s))\right.\nonumber\\&\hspace*{0.4in} \left. +\bar{\lambda}(s,X_s)\bar{f}(s,X_s,\Psi^{k-1}(t,X_s)))e^{-\int^s_tr(u,X_u)+\lambda(u,X_u)+\bar{\lambda}(u,X_u)du}ds\right.\nonumber\\&\hspace*{0.4in} \left. +e^{-\int^T_tr(s,X_s)+\lambda(s,X_s)+\bar{\lambda}(u,X_u)ds}\phi(X_T)\vert X_t=x\right], \forall k = 1,2,3...,
\end{align*}
where $\Psi^0(t,x) = \Psi_0(t,x).$\par
Then $\Psi^k(t,x)\ge \Psi^{k-1}(t,x), \forall k\ge 1, (t,x)\in [0,T]\times\mathbb{R}^m.$
\end{lemma}
\begin{proof}
Thanks to the monotonicity of $f(t,x,y)$ and $\bar{f}(t,x,y)$ in $y,$ and Lemma \ref{Lem:MonotonicityStep1Bilateral}, we have that
\begin{align*}
\Psi^{1}(t,x)\ge \Psi^0(t,x), (t,x)\in [0,T]\times\mathbb{R}^m.
\end{align*}
Then the conclusion can be obtained by the mathematical induction.
\end{proof}
\begin{lemma}\label{Lem:UpperBoundBilateral}
Let
\begin{align*}
I(t,x) &=  \mathbb{E}\left[\int^T_t( c(s,X_s)+\bar{\lambda}(s,X_s)(\bar{f}(s,X_s,V(s,X_s))-V(s,X_s)))e^{-\int^s_tr(u,X_u)du}ds\right.\nonumber\\& \hspace*{0.4in} \left.  +e^{-\int^T_tr(s,X_s)ds}\phi(X_T)\vert X_t=x\right].
\end{align*}
Then $\Psi^k(t,x)\le I(t,x), k=1,2,3..., (t,x)\in [0,T]\times\mathbb{R}^m.$
\end{lemma}
\begin{proof}
As $V$ has polynomial growth (see Theorem \ref{Main}), $I$ is well-defined. Then it follows from the Feynman-Kac formula that
\begin{align*}
\mathcal{L}I(t,x) - r(t,x)I(t,x) +\bar{\lambda}(t,x)(\bar{f}(t,x,V(t,x))-V(t,x)) +c(t,x) &= 0,\quad  (t,x)\in [0,T]\times\mathbb{R}^m,\\I(T,x) &= \phi(x), \quad x\in\mathbb{R}^m.
\end{align*}
Similarly, we have that
\begin{align*}
&\mathcal{L}\Psi^k(t,x)-(\lambda(t,x)+\bar{\lambda}(t,x)+r(t,x))\Psi^{k}(t,x)+\lambda(t,x)f(t,x,\Psi^{k-1}(t,x))\\&+\bar{\lambda}(t,x)\bar{f}(t,x,\Psi^{k-1}(t,x))+c(t,x)=0, \quad (t,x)\in[0,T)\times\mathbb{R}^m,\\&\Psi^k(T,x) = \phi(x), \quad x\in\mathbb{R}^m.
\end{align*}
As $f(t,x,y)\le y$, then Lemma \ref{Lem:MonotonicityBilateral} yields that
\begin{align}\label{L1}
f(t,x,\Psi^{k-1}(t,x)) \le \Psi^{k}(t,x).
\end{align}
Plug (\ref{L1}) into the PDE that $\Psi^k$ satisfies, then we have that
\begin{align}\label{L2}
\mathcal{L}\Psi^k(t,x)-(\bar{\lambda}(t,x)+r(t,x))\Psi^{k}(t,x)+\bar{\lambda}(t,x)\bar{f}(t,x,\Psi^{k-1}(t,x))+c(t,x)\ge 0.
\end{align}
Note that condition (\ref{MoralHazard}) implies that $y-\bar{f}(t,x,y)$ is monotone in $y$, then Lemma \ref{Lem:MonotonicityBilateral} and Lemma \ref{Lem:MonotonicityStep1Bilateral} yield that
\begin{align}\label{L3}
\Psi^k(t,x) - \bar{f}(t,x,\Psi^{k-1}(t,x))\ge \Psi^{k-1}(t,x) - \bar{f}(t,x,\Psi^{k-1}(t,x)) \ge  V(t,x) - \bar{f}(t,x,V(t,x)).
\end{align}
Then it follows from (\ref{L2}) and (\ref{L3}) that
\begin{align}\label{L4}
\mathcal{L}\Psi^k(t,x)-r(t,x)\Psi^{k}(t,x)+\bar{\lambda}(t,x)(\bar{f}(t,x,V(t,x))-V(t,x)) +c(t,x) \ge 0.
\end{align}
Let $Q(t,x) = I(t,x) - \Psi^k(t,x)$, then (\ref{L4}) and the PDE that $I$ satisfies yield that
\begin{align*}
\mathcal{L} Q(t,x)-r(t,x)Q(t,x)& \le 0,  \quad (t,x)\in[0,T)\times\mathbb{R}^m,\\&Q(T,x) = 0, \quad x\in\mathbb{R}^m.
\end{align*}
This implies $Q(t,x)\ge 0, (t,x)\in[0,T]\times\mathbb{R}$ and thus completing the proof.
\end{proof}


\bibliography{EbertReferences}        
\bibliographystyle{ecta}  

\end{document}